\documentclass{llncs}
\usepackage[T1]{fontenc}
\usepackage[utf8]{inputenc}
\usepackage{amsmath}
\usepackage{amsfonts}
\usepackage{amssymb}
\usepackage[english,ngerman]{babel}
\usepackage{algorithm2e}
\usepackage{tikz}
\usepackage{url}
\usepackage{ifthen}
\usepackage{subfigure}
\newcommand{\etal}{et al.\ }

\newcommand{\sodass}{\,:\,}
\newcommand{\setGilt}[2]{\left\{ #1\sodass #2\right\}}

\begin{document}
\title{Improved Fast Similarity Search in Dictionaries
\thanks{Partially supported by DFG Grant 933/5-1. A short version of the paper was presented at the 17th Symposium on String Processing and Information Retrieval (Spire 2010). The authors would like to thank Johannes Fischer for the fruitful discussions and the anonymous referees for the detailed reviews.}}
\author{Daniel Karch \and Dennis Luxen \and Peter Sanders}
\institute{Karlsruhe Institute of Technology\linebreak
\email{danielkarch@gmail.com, \{luxen, sanders\}@kit.edu }}
\selectlanguage{english}
\maketitle
\begin{abstract}
We engineer an algorithm to solve the approximate dictionary matching problem.
Given a list of words $\mathcal{W}$, maximum distance $d$ fixed at preprocessing time and a query word $q$, we would like to retrieve all words from $\mathcal{W}$ that can be transformed into $q$ with $d$ or less edit operations.
We present data structures that support fault tolerant queries by generating an index.
On top of that, we present a generalization of the method that eases memory consumption and preprocessing time significantly.
At the same time, running times of queries are virtually unaffected.
We are able to match in lists of hundreds of thousands of words and beyond within microseconds for reasonable distances.
\end{abstract}
\pagestyle{plain}
\section{Introduction and Previous Results}
The problem of searching approximate of matches in a dictionary arises in many fields.
Most common is the search for the so called best match. 
The problem has many applications.
For example, Google's 'Did you mean' feature catches typos in search queries. 
But in some settings, the uncertainty is higher and therefore one is not interested in the best match, but also in other matches which are still within a certain distance from the query. 
An interesting application is a geocoding application that maps perhaps misspelled locations descriptions to geocoordinates.

Each word is represented by a string of characters over a finite alphabet $\Sigma$. 
The Levenshtein distance \cite{levelshtein-66-binary} $ed(a,b)$ defines a metric between two words $a,b \in \Sigma^*$ and is used in this work to compute the distance between two words.

The most trivial algorithm to solve the problem is scanning sequentially through the input list and noting the best match(es) at each entry. 
The running time is obvious and consists of a linear number of distance computations and searching an entire directory on a standard desktop computer takes only a few seconds even for dictionaries up to a few hundred thousand or even a million entries. 
But in many settings this is too much, because queries arrive in a high frequency.
For example, a web search engine only has a few milliseconds to process a single request and does not have the time to do exhaustive searching in a large dictionary.

Since these distance computations are rather expensive, it is natural to find an algorithm that does not compare the input to the entire dictionary, but only a few entries.
A so-called \textit{filter} represents a criterion to quickly discard large portions of the search space. 

The exploitation of the underlying metric space implied by the edit distance \cite{BaezaYates1998} is easy. 
The set of words is partitioned by the distance of each element to a more or less carefully chosen and perhaps random pivot element. 
By computing the distance to the pivot, the search space is pruned using the triangle inequality. 
However, this approach has limited effect, e.g. in natural language dictionaries.
Distances of most dictionary elements to the pivot lie in a small range and pruning has a limited effect.

To cope with the limitations, different schemes were introduced from using multiple pivots to tree-like data structures. 
The oldest of such trees is the BK-tree data structure proposed by Burkhard and Keller \cite{362025}, which is built recursively.
A root is selected whose subtrees are identified by distance values to the root. 
The $i$-th subtree consists of elements of the dictionary at distance $i$ to the root. 
The subtrees are recursively built until the number of elements in a subtree is below some threshold. 
Again, the triangle inequality is used to branch into or cut any subtrees. 
A candidate set of possible matches is built by the union of all leaves that are reached by the tree traversal. 
A rather weak result is that BK-trees and its refinements need $O(n^\alpha)$, $0<\alpha <1$, comparisons and node traversals on average \cite{BaezaYates1998} for a dictionary of $n$ entries. 
See Ch\'{a}vez \etal's publication \cite{502808} for a survey. 

The general problem of approximately matching words can be further refined into two categories, namely matching elements from a set of words or matching arbitrary patterns in strings \cite{BaezaYates1998}. 
As usual, in high dimensional search problems there is a severe space-time trade-off. 
Cole \etal \cite{CGL04} give a solution for the dictionary matching problem using $O(n\log^dn)$ space and answer a query in $O(m\cdot\log\log n+\mathit{occ})$ for a dictionary of size $n$, query length $m$, edit distance $d$.
Here, $\mathit{occ}$ is the number of occurrences of the pattern. 
Mihov and Schulz \cite{1105590} present a sophisticated but complicated method to solve the problem with universal Levenshtein automata.
Russo \etal \cite{Russo2007} propose a compressed index that performs well for $d=1,2,3$, but needs several seconds to perform queries for larger $d$.
The best known linear space solution needs $O(m^{d-1}\log n\log\log n + \mathit{occ})$ query time \cite{CLSTW06} for error $d\geq$ 2.
However, this solution is fairly complicated and involves large constant factors, and to our knowledge there aren't any implementations yet.
Furthermore, any of the general-purpose approximate string matching algorithms have to be adapted to perform dictionary matching:
Either the query has to be adapted to ensure that only complete words are found, or special characters have to be introduced to mark the start and end of a dictionary entry.

More practically oriented work has focused on filtering algorithms that take linear space, but these do not have strong worst case performance guarantees.
K{\"a}rkk{\"a}inen and Na \cite{DBLP:conf/alenex/KarkkainenN07,135908} report on a linear space data structure that supports substring search, but has much larger query times compared to our result.
Ukkonnen \cite{Ukkonen93approximatestring} investigated suffix trees as a building block to solve the problem. 
Likewise, Cobbs \cite{Cobbs1995} gives a data structure based on suffix trees with linear time preprocessing for a fixed size alphabet for searching fixed patterns.
Queries to the data structure can be answered in time $O(mq+occ)$, where $m$ is the length of the pattern, $q\leq n$ and again $occ$ is the number of occurrences.

A technique involving so called $q$-grams is popular among practitioners. 
But it generally works for the Hamming distance only.
$q$-grams are sub-words of length $q$ and the $q$-gram distance (or similarity) is defined by the number of $q$-grams two words share. 
A generalization of this technique are gapped $q$-grams. 
Taking $q$ letters from a word as before and introducing \textit{don't care} defines a pattern instead of sub-word. 
These don't care positions are then called gaps. 
In \cite{DBLP:conf/cpm/BurkhardtK02} it is shown that one-gapped $q$-grams can be extended to obey the edit distance metric. 
One of the major difficulties of gapped $q$-grams is the computation of a threshold which is the smallest number of matching $q$-grams between a pattern and a text. 
Most experimental work focuses on finding this threshold, e.g. \cite{672762,Burkhardt01betterfiltering}.

For more information on approximate string matching see \cite{DBLP:conf/alenex/KarkkainenN07,Maass2007662,1119494,1183723}.

To speed up edit distance computation itself, research focused on simple and practical bit-vector algorithms \cite{citeulike:556192}.
Words of character length $n$ with $d$ or fewer differences can be matched in $O(nmd/w)$, where $w$ is the word size of the machine an $m$ the length of a query. 
This is done by computing the bit representation of the current state-set of the $k$-difference automaton. 
The running time was further improved to $(nm/w)$ \cite{316550} and further refinements \cite{316550} yield an $O(dn/w)$ expected-time algorithm for arbitrary large $m$.

The remaining parts of this paper are structured as follows. 
Section \ref{sec:apxmatching} gives an introduction into the neighborhood relation on strings that we exploit. 
It is followed by an discussion of our experimental results in Section \ref{sec:experimental}.
Finally, Section \ref{sec:conclusions} draws conclusions and identifies future work.
\section{Approximate Dictionary Matching}\label{sec:apxmatching}

Our method can be seen as an implementation of a general approach to approximate matching known as \emph{(lossless) filtering}. This can be formalized as follows: Given a set $\mathcal{S}$ of words over a finite alphabet $\Sigma$, a metric $\delta:\Sigma^*\times\Sigma^*\rightarrow\mathbb{R}_0$, and an error threshold $d$, 
a preprocessing algorithm produces a data structure that allows fast evaluation of a function $F:\Sigma^*\rightarrow \mathcal{P(S)}$. For a query word $q\in\Sigma^*$, $F(q)$ computes a set of candidate words from $\mathcal{S}$
such that the set of approximate matches $\setGilt{s\in\mathcal{S}}{\delta(q,s)\leq d}$ is a subset of $F(q)$.

\paragraph{Deletion Neighborhood.}

We improve a filtering technique called \textit{Fast Similarity Search (FastSS)} \cite{BoHuSt07} which is a generalization of a single error
method proposed by Mor and Fraenkel \cite{358752}.

For integer $d$ and a word $w\in \Sigma^*$ the 
\emph{$d$-(deletion-)neighborhood} 
$\mathcal{N}_d(w)$ is defined as the set of all subwords of $w$ with exactly $d$ deleted positions.
Each element of $\mathcal{N}_d(w)$ is called a \textit{residual string}. 
Furthermore, a string $w$ is called originating string for residual $r$ if and only if $r\in \mathcal{N}_d(w)$. We obtain a lossless filter
for a set of words $\mathcal{S}$ by precomputing the $d$-neighborhoods
of strings in $\mathcal{S}$. As a filtering function, we obtain $F(q)=\setGilt{s\in\mathcal{S}}{\mathcal{N}_d(s)\cap\mathcal{N}_d(q)\neq\emptyset}$.

The correctness of this definition follows from the following Lemma:

\begin{lemma}
\label{obsResidual}
If two words $u, v \in \Sigma^*$ are within a distance $d$ from each other, then there exists a word $w$ which has length at 
least $\vert u\vert-d$ and consists of letters from $u$ and $v$ in their original order. 
Assume that $u$ is at least as long as $v$.
\end{lemma}

We use the concept of \textit{Ordered Edit Sequences} \cite{Maass2007662} to show the claim. Our proof is simpler and more intuitive than the proof from \cite{BoHuSt07}.
\begin{proof}
Recall that the edit distance is said to be the minimal number of edit operations to transform one word $u\in\Sigma^*$ into another $v\in\Sigma^*$. 
The set of operations available for any single transformation are $op=\mathsf{\{ins, del, chg\}} : \Sigma \cup\{\epsilon\}\rightarrow\Sigma\cup\{\epsilon\}$ with $v=op_d(op_{d-1}(\ldots(op_1(u))\ldots))$. 
The sequence $\rho(u,v)=
(op_1,op_2,\ldots,op_d)$ is called edit sequence and we call it ordered if the
operations are applied from left to right. 
We define $pos(\cdot)$ to give the position of an operation within the edit sequence.
In  other words $\forall i:\left(pos(op_i)\leq pos(op_{i+1})\right)$. 
By definition of the edit distance metric there exists an edit sequence of minimal length.
Now, we can show Lemma~\ref{obsResidual}. 
Since $ed(u,v) \leq d$ it follows that the length of a minimal ordered edit  sequence is at most $d$, which means $\vert \rho_{min}(u,v)\vert\leq d$ is the length of a minimal edit sequence.
This implies that $v$ is changed at no more than $d$ positions. 
By deleting these at most $d$ positions from $v$, we get a string $w$, which has length at least $\vert u \vert - d$ and preserves the letter ordering from $u$ and $v$.
\hfill$\square$
\end{proof}

\paragraph{Basic Data Structure.}
A static index data structure is generated in a precomputation phase that can be queried during an on-line phase.
We insert a number of values into a hash table that is part of our data structure.
The structure utilizes the hash table to store pointers to originating dictionary entries at the hash values of residual strings.
If any hash value has more than one originating dictionary entry then the corresponding pointers are stored in a list. 
Figure \ref{fig:datastructure} sketches the internal structure of the index.
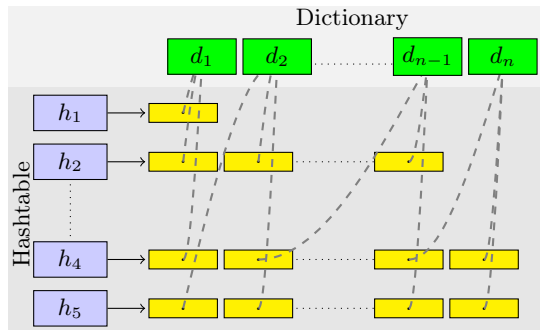
\begin{figure}
\begin{center}
\begin{tikzpicture}

\tikzstyle{block1} = [draw,fill=blue!20,minimum width=3em]
\tikzstyle{block2} = [draw,fill=green,minimum width=2.8em]
\tikzstyle{block3} = [draw,fill=yellow,minimum width=2.8em]
\tikzstyle{pointer}= [->,thick,gray,dashed,step=1pt]

\node[minimum height=3.5em,minimum width=22em,fill=gray!10,overlay] at (4.75,-0.8) {};
\node[minimum height=10.1em,minimum width=22em,fill=gray!20,overlay] at (4.75,-2.95) {};

\draw[dotted] (2,-4*0.65) -- (2,-5.5*0.65);
\draw[dotted] (5,-1) -- (7,-1);
\draw[dotted] (4.5,-2*0.65-1) -- (6.5,-2*0.65-1);
\draw[dotted] (4.5,-4*0.65-1) -- (6.5,-4*0.65-1);
\draw[dotted] (4.5,-5*0.65-1) -- (6.5,-5*0.65-1);
\node[rotate=90] at (1.35,-3) {Hashtable};
\node at (5.75,-0.4) {Dictionary};

\foreach \x/\y in {1/a,2/b,4/n-1,5/n} {
 \ifthenelse{\x > 3}
   { \node[block1] at (2,-\x*0.65-1) (block\x) {$h_\x$}; }
   { \node[block1] at (2,-\x*0.65-1) (block\x) {$h_\x$}; }
 \draw[->] (block\x.east) -- +(0.5,0);
 \node[block3] at (\x+2.5, -4*0.65-1) {.};
 \node[block3] at (\x+2.5, -5*0.65-1) {.};
 \ifthenelse{\x < 2}
 {\node[block3] at (\x+2.5, -1*0.65-1) {.}; } {}
 \ifthenelse{\x < 5}
 {\node[block3] at (\x+2.5, -2*0.65-1) {.}; } {}
}

\draw[pointer] (1+2.5, -1*0.65-1) sin (1+3-0.25, -1);

\draw[pointer] (1+2.5, -2*0.65-1) sin (1+3-0.25, -1);
\draw[pointer] (2+2.5, -2*0.65-1) sin (2+3-0.25, -1);
\draw[pointer] (4+2.5, -2*0.65-1) cos (4+3-0.25, -1);

\draw[pointer] (1+2.5, -4*0.65-1) cos (1+3-0.25, -1);
\draw[pointer] (2+2.5, -4*0.65-1) cos (4+3-0.25, -1);
\draw[pointer] (4+2.5, -4*0.65-1) cos (5+3-0.25, -1);
\draw[pointer] (5+2.5, -4*0.65-1) cos (5+3-0.25, -1);

\draw[pointer] (1+2.5, -5*0.65-1) sin (2+3-0.25, -1);
\draw[pointer] (2+2.5, -5*0.65-1) cos (2+3-0.25, -1);
\draw[pointer] (4+2.5, -5*0.65-1) cos (4+3-0.25, -1);
\draw[pointer] (5+2.5, -5*0.65-1) cos (5+3-0.25, -1);

\foreach \x/\y in {1/a,2/b,4/n-1,5/n} {
 \ifthenelse{\x > 3}
  {\node[block2] at (\x+2.75, -0.9) (dict\x) {$d_{\y}$};}
  {\node[block2] at (\x+2.75, -0.9) (dict\x) {$d_\x$};}
}

\end{tikzpicture}\end{center}
\caption{Approximate string matching data structure.}
\label{fig:datastructure}
\vspace{-2em}
\end{figure}\paragraph{Query.}
For an input query $q$ and maximum distance $d$, the corresponding $d$-neighborhood and its hash values are computed. 
If any element of the query's residuals is also an element of the data structure then the pointers to the originating dictionary entries give a set of candidates.
Each of those might be an approximate match.
Once the candidate set is completely built, it is searched exhaustively by computing the edit distance of each candidate to the query.
By removing all elements from the candidate set whose distance is larger than the threshold $d$ we get the set of all dictionary members that are at most a distance $d$ away from query $q$. 
Perhaps there exists an additional order on the candidates stemming from the application.
The algorithm can be adapted to not only return the best match, but also a list of those candidates that are sufficiently close.

\paragraph{Precomputation.}
We compute the $d$-neighborhood of each element of the input dictionary and insert the resulting information into our index data structure.
Doing this precomputation naively and storing all residual strings in a data structure takes up an enormous amount of space. 
Instead, we use hashing and reduce each element of the residual neighborhood into an integer number. 
We insert pointers to the originating dictionary entries into the hash table at the respective hash values of all residual strings.
Therefore, only constant space is needed per residual string regardless of the length of that string.
We now present an improvement to the algorithm.

\paragraph{Algorithmic Generalization.}
We limit the number of elements that are inserted into the index while staying lossless.
To do so, we split long input words in half, compute the residual strings with half the number of errors, and adapt the query algorithm, which will be explained in this Section.
See Section \ref{sec:experimental} for an analysis of the threshold value $m$, which indicates whether or not to split a word.
Instead of generating $\vert s \vert \choose d$ hash values we insert only $$ {\vert s \vert \choose \lfloor \frac{d}{2} \rfloor } + {\vert s \vert \choose \lceil \frac{d}{2} \rceil }$$ values for a split dictionary entry $s$. 
The \textit{generalized $d$-neighborhood} of $w'\in\Sigma^*$ is the set of residuals that is found by computing all combinations of $\lceil \frac{d}{2} \rceil$ deleted characters for the first and second half of $w'$.

The generation of the index is simple.
But we have to pay some extra care at query time, because insertions and deletions that transform words $w$ into $w'$ can take place at arbitrary positions.
As a consequence, we can not rely on the length of a query $q$ to decide whether it has been split or not.
Instead of splitting a query $q$ of length $l$ at a fixed position, it is split several times in half at positions in the interval of $\lceil\frac{l}{2}\rceil \pm \lceil \frac{d}{2} \rceil$. 
Also, the allowed error is halved.
If the length of an input word is within $m\pm d$ then the index is also searched for the non-split string.

Consider these definitions.
Let $ w\in \Sigma^*$ be an entry of dictionary $D$ and $d$ the maximum allowed error. 
Let $u=p(w)$ and $v=s(w)$ denote the first and second half of the split word $w$. 
Prefixes $u$ and suffix $v$ are indexed, while $q$ is the query.
Any query $q$ is split at several positions as explained above and we define $\mathcal{P}(w)$ to be the set of first and $\mathcal{S}(w)$ to be the set of second halves.
Our method is still correct since we can show the existence of a common residual string for either the prefix or the suffix of a split query word by the following Lemma.
\begin{lemma}\label{lemma:halved}
Let $q \in \Sigma^*, w=uv$ with $ed(w,q) \leq d$. 
Consider $\mathcal{P}(q)$ $\left(\mathcal{S}(q)\right)$ to be the set of $\lceil\frac{d}{2}\rceil$ many prefixes (suffixes) of $q$ that are generated for each query to the index.
Then there exists at least one pair $(p',s')$ with $p'\in \mathcal{P}(q)$, $s'\in \mathcal{S}(q)$ and $p'\circ s'=q$ of prefix-suffix-elements for which either $ed\left(u,p'\right) \leq \lceil d/2 \rceil$ or $ed\left(v,s'\right) \leq \lceil d/2\rceil$.
It suffices to test the split positions from the interval $\lceil\frac{\vert q\vert}{2}\rceil\pm\lceil\frac{\vert d \vert}{2}\rceil$ to find that pair.
\end{lemma}
\begin{proof}
Consider the edit sequence $S$ that transforms $w$ into $q$ and that has length at most $d$, s.t. $ed(w,q)\leq d$.
String $w$ is split at position $\lceil\frac{\vert w\vert}{2}\rceil$ into $w=p\circ s$.
Note that the lengths of $p$ and $s$ differ at most $1$.
Sequence $S$ is applied to $w=p\circ s$ and yields $q=p'\circ s'$.
Hence, either $ed(p,p')\leq \lceil\frac{d}{2}\rceil$ or $ed(s,s')\leq \lceil\frac{d}{2}\rceil$ or both.
The algorithm has to split query $q$ exactly into $p'$ and $s'$ to guarentee that a match is found.
Assume that it doesn't suffice to test the interval $\lceil\frac{\vert q\vert}{2}\rceil\pm\lceil\frac{\vert d \vert}{2}\rceil$ to find the correct splitting position.
Then $p'$ is either shorter than $\lceil\frac{m}{2}\rceil-\lceil\frac{d}{2}\rceil$ or longer than $\lceil\frac{m}{2}\rceil+\lceil\frac{d}{2}\rceil$.
Assume $\vert p'\vert < \lceil\frac{m}{2}\rceil-\lceil\frac{d}{2}\rceil$. Then

\noindent $\Rightarrow\vert s'\vert>\lceil\frac{m}{2}\rceil+\lceil\frac{d}{2}\rceil$\\
$\Rightarrow \vert p' \vert + \lceil\frac{d}{2}\rceil<\frac{m}{2}<\vert s'\vert - \lceil\frac{d}{2}\rceil$\\
$\Leftrightarrow\vert p' \vert + \lceil\frac{d}{2}\rceil<\vert s'\vert - \frac{d}{2}$
$\Leftrightarrow\vert p'\vert < \vert s'\vert -2\cdot\lceil\frac{d}{2}\rceil$
$\Leftrightarrow\vert p' \vert - \vert s'\vert < -2\cdot\lceil\frac{d}{2}\rceil$\\
$\Leftrightarrow\vert s'\vert-\vert p'\vert> 2\cdot\lceil\frac{d}{2}\rceil$

\noindent This implies that the lengths of $s'$ and $p'$ differ by more than $2\cdot\lceil\frac{d}{2}\rceil$.
But then edit sequence $S$ has to be longer than $2\cdot\lceil\frac{d}{2}\rceil$ operations, because length difference is a lower bound for edit distance.
The other case for $\vert p'\vert > \lceil\frac{m}{2}\rceil+\lceil\frac{d}{2}\rceil$ follows by the same line of argumentation.\hfill\qed
\end{proof}
Wu and Manber \cite{135244} use partitioning into $d+1$ pieces to match one of the pieces with no error, while Navarro and Baeza-Yates \cite{Navarro98improvingan} gave a recursive partitioning scheme for fast on-line approximate string matching.

See Section \ref{sec:experimental} for an experimental analysis of the generalization that shows it uses half the space than our implementation of the original algorithm and maintains stable query performance.

\section{Analysis}\label{chap:analysis}
Our variant makes heavy use of hashing as we argued before and we analyze the penalty of our approach coming from hash collisions.
First, consider the case that we do not split the input string, which resembles the original method.

For each dictionary entry of length $\ell$, we insert at most $\ell\choose d$
constant size entries into the hash table. The hash table needs $O(1)$ space per
element since the bit size of each entry is of constant size. 
Note that for $d=1$ we obtain overall linear space
because $O(\ell)$ constant size hash table entries are stored for
a dictionary entry of size $\ell$.  

We resort to average case analysis for the query time using the following
model: 
Consider a dictionary of $n$ words drawn uniformly at random from 
$\Sigma^{\ell}$ and an arbitrary query word $q$ of length $\ell$.
In real world inputs, we have a mix of words with different lengths.
However, a query of length $\ell$ will mostly return candidates of length $\ell$
for random inputs. Hence, there is no need to postulate anything on the
distribution of lengths -- we just analyze the system for each length
separately. 

Assume an order in which the residuals of a word can be generated.
Consider the 0/1 random variable $X_{ijk}$ that has value one iff
the $i$-th residual of query $q$ is equal to the $j$-th residual of the input word
$k$. The total number of residuals that need 
to be considered is
bounded by 
$$X:=\sum_{i=1}^{\ell\choose d}\sum_{j=1}^{\ell\choose
d}\sum_{k=1}^{n}X_{ijk}.$$
This is an overestimation of the actual number of residuals to be
considered since by deleting different sets of 
characters we might arrive 
at the same residual. However, for not too small $\Sigma$ this only happens rarely
Let $\sigma$ denote the size of the alphabet actually used. 
We have $P[X_{ijk}=1]=1/\sigma^{\ell-d}=
\sigma^{d-\ell}$. Hence, using the linearity of expectation, we get an expected
value of 
\begin{equation}\label{eq:time}
\text{\bf E}[X]=n{\ell\choose d}^2\sigma^{d-\ell}
\end{equation} 
This gives the number of residuals we have to consider.
The number of actual distance computations may be smaller since several residuals of $q$ may match several residuals of a dictionary entry $s_k$, but we will compute the distance $d(s_k,q)$ only once.

An interesting consequence of (\ref{eq:time}) is that, on average, we can expect a speedup over the naive algorithm that is independent of the size of the input dictionary. 
By applying the Markov inequality, we can estimate an upper bound of the probability that the expected number is not a fraction of $n$. 
Let $c$ be a constant $>0$.
\begin{equation}
\label{eqn:markoff}
P\left[X\geq \frac{n}{c}\right]\leq c^{-1}{\ell\choose d}^2\sigma^{d-\ell}\text{~.}
\end{equation} 
See Section \ref{sec:experimental}, where we experimentally analyze the behavior of the algorithm for varying splitting parameters.

\section{Experimental Results}\label{sec:experimental}
\paragraph{Implementation Details.}
We implemented the data structure, the construction and query algorithms in C++ using GCC Compiler version 4.3.2.
We hashed all residual strings with the built-in hash function of the Boost library v$1.36$ to a $32$-Bit Integer and chained with a simple linear congruence.

The exhaustive search of the candidate set is done by a simple implementation of the Levenshtein distance. 
It computes a band of width $2d+1$ only. 
This way we compute the distance exactly only if it is smaller than $d$ and return otherwise as soon as we get a certificate that the distance is larger than $d$. 
Since we need $O(1)$ to fill a cell in the distance table, we can verify a candidate in $O(d\cdot l)$, where $l$ is the length of the shorter word. 
In the experiments it took less than a microsecond to verify any single candidate.
\paragraph{Environment.}
All of our tests were conducted on a single core of a Intel Xeon X5550 CPU, running a version 2.6.27 Linux kernel.
We compare the performance of our optimizations against our own implementation only for reasons of fairness.
\paragraph{Test Instances.}
The sizes of the dictionaries used in the experiments range between about $38\,000$ and $1.8$ million entries (see Table~\ref{tab:size}). 
All results were averaged over a number of queries of perturbed dictionary entries.
\begin{table}[t]
\begin{center}
\begin{tabular}{l|r|r|r}
\label{tab:size}
 dictionary & no. elements & avg. length & size [MiB]\\ 
\hline mobydick  &     37\,924 &  9 &  0.31\\
\hline town      &     47\,339 & 10 &  0.49\\ 
\hline english   &    213\,557 & 10 &  2.20\\ 
\hline wikipedia & 1\,812\,365 &  9 & 17.06\\
\end{tabular}
\end{center}
\caption{Basic information on our dictionaries.}
\end{table}
The word list \textit{mobydick} consists of the distinct words from Melvilles
classic novel, the \textit{town} dictionary consists of German 
town names extracted from the OpenStreetMap
project\footnote{\url{http://www.openstreetmap.org/}} in February 2009, the \textit{english} 
dictionary is an extract of words from Webster's English Dictionary and the
\textit{wikipedia} dictionary is the list of 
pairwise distinct words from all english Wikipedia\footnote{\url{http://www.wikipedia.org}} titles as of February 2009.
Table \ref{tab:size} lists element count and average word length of each test data set.

\subsection{Splitting Parameter}
\paragraph{Preprocessing Space.}
We analyze the amount of distinct residuals that are generated for each value of $m\in 1,\ldots,30$ and the average duration of a single query against this index. 
To do so, we averaged over $1\,000$ randomized queries.
Both value $m=1$ and $m=30$ resemble worst cases. 
We present the results in the plots of Figure \ref{fig:numOfResids} for edit distance $3$.
Other distances show similar behavior. 
Note that we omitted the lower and upper values of $m$ for clearer arrangement, because for the values $1,\ldots,5$ ($20,\ldots,30$) nearly all (none) strings get split.
We present selected plots that show the experiments.
Note the logarithmic scales for query times.
\begin{figure}[t]
\center\scalebox{0.48}{\rotatebox{270}{\includegraphics{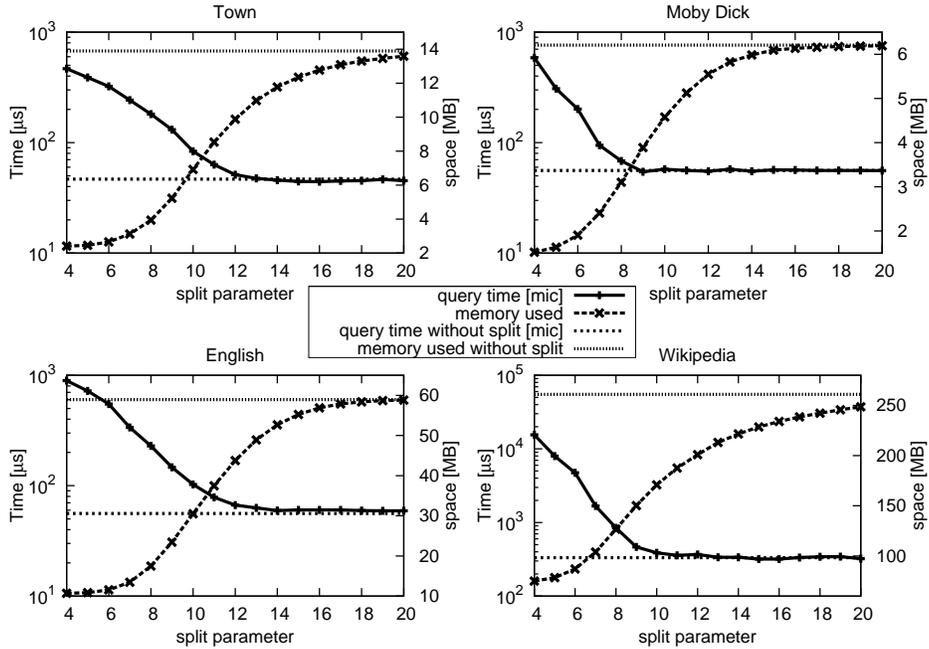}}}
\caption{Analysis of the Splitting Parameter for $d=2$}\label{fig:numOfResids}
\end{figure}
In all the experiments we see that there is a trade-off between the memory consumption and average query time.
The split parameter functions as an adjusting value to choose between size of the index and query performance. 
Our analysis shows that the index size can be halved by degrading the speed of an average query within acceptable limit only.
Especially, when splitting is restricted to those dictionary entries whose length is larger than the average, we can halve the memory consumption of the index.
The query performance is virtually unaffected.
\paragraph{Preprocessing Time.} We investigated preprocessing times with and without splitting parameter set. 
The preprocessing was run for values $d=0,\ldots,4$ on all of our data sets. Figure \ref{fig:preproc} reports on the numbers.

\begin{figure}[t]
\center\scalebox{0.48}{\rotatebox{270}{\includegraphics{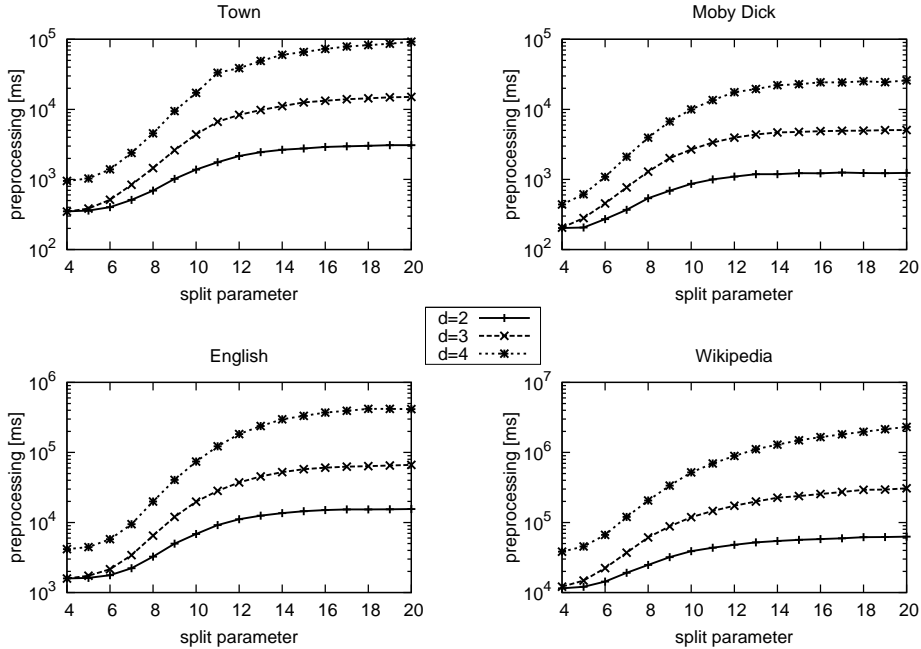}}}
\caption{Analysis of the preprocessing in relation to the splitting parameter.}\label{fig:preproc}
\end{figure}

The preprocessing is roughly ten times faster for reasonable values of the splitting parameter than without any splitting.
Mainly this is because we do not store any additional information besides pointers to dictionary entries.

\paragraph{Query Performance.}
We conducted experiments on each list for maximum distances of $d=\{0,\ldots,4\}$ to test the query performance for varying number of allowed errors. 
For natural language dictionaries a distance of $d=3$ is already large and larger distances deliver matches that already look arbitrary.
During each query we generated the candidate set, verified each member of the set and reported a best match found. 
Each test run picked $1\,000$ elements from the dictionary and introduced up to $d$ errors at random. 
The splitting parameter is set to $m = 10$.
The query times and search space sizes are averaged. 
Tables \ref{tab:proc} and \ref{tab:nums} report on these experiments.
\begin{table}[t]
\begin{center}
\begin{tabular}{c|r|r||r|r||r|r||r|r}
 & \multicolumn{2}{c||}{Mobydick} & \multicolumn{2}{c||}{Town} &
\multicolumn{2}{c||}{English} & \multicolumn{2}{c}{Wikipedia}\\
     $d$ & mem   & proc   & mem   & proc   & mem    & proc    & mem     & proc    \\ 
\hline 0 &  0.25 &  0.061 &  0.46 &  0.156 &   2.36 &   0.886 &   14.41 &   7.131 \\
\hline 1 &  1.33 &  0.320 &  1.79 &  0.576 &   8.55 &   3.450 &   55.84 &  32.287 \\
\hline 2 &  4.57 &  1.272 &  6.91 &  2.483 &  30.49 &  12.596 &  170.79 & 107.289 \\ 
\hline 3 &  9.78 &  4.044 & 15.18 &  7.458 &  61.37 &  36.309 &  342.18 & 270.506 \\ 
\hline 4 & 16.09 & 14.647 & 27.20 & 28.144 & 105.75 & 117.970 &  603.35 & 922.521 \\ 
\end{tabular}
\end{center}
\caption{\textbf{Preprocessing:} \textit{Mem} is the size of the index in [MiB], \textit{proc} the duration $[s]$.}
\label{tab:proc}
\end{table}

\begin{table}[t]
\begin{center}
\begin{tabular}{c|r|r||r|r||r|r||r|r}
 & \multicolumn{2}{c||}{Mobydick} & \multicolumn{2}{c||}{Town} & \multicolumn{2}{c||}{English} & \multicolumn{2}{c}{Wikipedia}\\
   & query     & cand  &  query    & cand  &  query    & cand & query    & cand \\ 
 $d$ & [$\mu$s]  & set   &  [$\mu$s] & set   &  [$\mu$s] & set  & [$\mu$s] & set  \\
\hline 0 &      2 &      1 &      0 &      2 &      0 &      1 &           1 &            1\\
\hline 1 &      5 &      5 &      8 &      9 &      8 &      6 &          34 &           25\\
\hline 2 &     84 &     61 &     99 &     99 &    122 &     46 &         502 &          702\\ 
\hline 3 &    553 &    606 &    644 &    613 &    644 &    502 &      7\,019 &       9\,900\\ 
\hline 4 & 2\,974 & 3\,376 & 7\,250 & 3\,720 & 7\,250 & 4\,520 & 55$\cdot 10^3$ & 65$\cdot 10^3$\\ 
\end{tabular} 

\end{center}
\caption{\textbf{Query:} \textit{query} is the average time for a single query in microseconds, \textit{cand set} the average cardinality of the candidate set.}
\label{tab:nums}
\end{table}

The \textit{query} column shows the time for the actual query in microseconds and \textit{cand set} is the number of elements in the candidate set on average. 
We see the expected rise in the number of candidates that have to be verified by the algorithm. 
We briefly compared the observed number of collisions against the expected number from our analysis in Section \ref{chap:analysis}. 
The observed number was always lower as the expected one since our analysis is an overestimate of the actual collision rate. 
In some cases we observed the order of a magnitude less collisions than expected.

When looking at our result and the original experiments of Bocek \etal \cite{BoHuSt07} in Table \ref{tab:compare} we see that our implementation performs better by about an order of magnitude in all important areas.
Although we know that our numbers were measured on different hardware, they give an impression on the performance.
The experiments were run on the same random dictionary of 10\,000 words.
Note that the case of $m=\infty$ corresponds to Bocek \etal's algorithm.
They proposed several improvements that either perform fast or have low space consumption but not both at the same time.
Since the results of the experiments are only available as plots we have to estimate the values.
We did so in a benevolent way and compare the best of their values in each category against our implementation with and without splitting. 
\begin{table}[!h]
\begin{center}
\begin{tabular}{r|r|r|r||r}
                      & $m=\infty$ & $m=10$          & Best of Bocek \etal  & BK-tree  \\\hline
preprocessing [ms]    & 2\,649     & \textbf{349}    & 5\,000 - 7500        & \textbf{183}\\
avg. query [$\mu$s]   & 114        & \textbf{18}     & 100--200$\cdot 10^3$ & 935\\
dictionary size [MiB] & 9.8        & \textbf{1.5}    & 20                   & \textbf{0.25}
\end{tabular}
\end{center}
\caption{Comparison Against Existing Experiments, best results bold and BK-tree for reference.}\label{tab:compare}
\end{table}
We see one potential source of performance problems with our experiments as we tested on dictionaries with rather short words that have similar sizes. 
The higher the allowed error distance $d$ is, the shorter residual strings get. 
This leads to longer indices lists in the hash table, because it is more  likely that two distinct words will have common residual strings. 
This also explains the larger number of candidates for higher values of $d$. 

An experimental evaluation of BK-trees \cite{DBLP:journals/corr/cs-DB-0403014} and several variants reports on the size of the search space that is visited depending on the allowed error distance. 
Those experiments were conducted on a set of 100\,000 English words and report on a nearly linear growth of the visited search space going up from $5\%$ for edit distance 0 to slightly more than $40\%$ for a distance of 4.
The size of the visited search space in our experiments is always less than $1\%$ and much less than the search space size for the best BK-tree variant \cite{DBLP:journals/corr/cs-DB-0403014}.
We were able to confirm the high number of candidates with our own BK-tree implementation.
Table \ref{tab:bkexp} reports on selected numbers of those experiments for the largest and smallest of the dictionaries.
\begin{table}
\begin{center}
 \begin{tabular}{c|r|r||r|r}
  & \multicolumn{2}{c||}{Mobydick} & \multicolumn{2}{c}{Wikipedia}\\
    & query     & cand  & query    & cand \\ 
  $d$ & [$\mu$s]  & set   & [$\mu$s] & set  \\
 \hline 1 &     198 &            197 &          1\,258 &          1\,184\\
 \hline 2 &  3\,586 &         4\,127 &  94$\cdot 10^3$ & 116$\cdot 10^3$\\ 
 \hline 3 &  8\,722 & 10$\cdot 10^3$ & 374$\cdot 10^3$ & 486$\cdot 10^3$\\ 
 \hline 4 & 13\,083 & 15$\cdot 10^3$ & 862$\cdot 10^3$ & 802$\cdot 10^3$\\ 
 \end{tabular} 
\end{center}
\caption{Selected numbers on the performance of BK-trees.}
\label{tab:bkexp}
\end{table}

The number of candidates in BK-trees is high even for small allowed error distances.
Thus the filtering effect of the metric space is quite low.
\section{Conclusions and future work}\label{sec:conclusions}
We improved a method for approximate string matching in a dictionary. 
We developed algorithmic optimizations that provide a tuning parameter to choose between space consumption and running time while having overall lower preprocessing duration.
Additionally, the performance has been validated experimentally by comparison against BK-trees and the baseline version of FastSS.

We see possibilities to speed up the verification of the candidate set using bit-parallelism \cite{DBLP:conf/wea/HyyroFN04} and SIMD instructions of current processors. 
This technique has been successfully used by \cite{1221301}. 
However, only about half of the time of the algorithm is actually spent in the verification phase with the computation of the edit distance.
Likewise there might be opportunities to speed up the precomputation, in particular, using fast, incremental computations of hash functions and using parallelization. 
On the other hand, it might be interesting to use data compression techniques to further reduce the storage requirements.
\bibliography{improvFastSS}
\bibliographystyle{splncs}
\end{document}